\documentclass[sigconf]{acmart}
\AtBeginDocument{%
  }

\usepackage{booktabs}
\usepackage{multirow}
\usepackage{graphicx}
\usepackage{enumitem}
\usepackage{tabularx}

\newtheorem{lemma}{Lemma}

\newcommand{\ie}{\emph{i.e., }}
\newcommand{\eg}{\emph{e.g., }}

\newcommand{\cf}{\emph{cf. }}
\newcommand{\aka}{\emph{a.k.a. }}

\definecolor{myred}{HTML}{ffdfdf}

\setcopyright{acmlicensed}
\copyrightyear{2026}
\acmYear{2026}
\acmDOI{XXXXXXX.XXXXXXX}

\acmConference[Preprint]{arXiv Preprint}{2026}{Hangzhou, China}
\acmISBN{XXX-X-XXXX-XXXX-X/XXXX/XX}

\usepackage{multirow}
\usepackage[table,xcdraw]{xcolor}
\usepackage[normalem]{ulem}
\usepackage{adjustbox}
\usepackage{enumitem}
\usepackage{subcaption}
\begin{document}

\title{LBR: Towards Mitigating Length Bias in Large Language Models for Recommendation}

\author{Hongchen Li}
\email{void_jack@zju.edu.cn}
\orcid{0009-0006-2317-8305}
\authornotemark[2]
\authornotemark[3]
\affiliation{%
  \institution{Zhejiang University}
  \city{Hangzhou}
  \state{Zhejiang}
  \country{China}
}

\author{Bohao Wang}
\orcid{0009-0006-8264-3182}
\email{bohao.wang@zju.edu.cn}
\authornotemark[2]
\authornotemark[3]
\affiliation{%
  \institution{Zhejiang University}
  \city{Hangzhou}
  \state{Zhejiang}
  \country{China}
}

\author{Jingbang Chen}
\orcid{0000-0002-7279-0801}
\email{chenjb@cuhk.edu.cn}
\affiliation{%
  \institution{The Chinese University of Hong Kong}
  \city{Hong Kong}
  \country{China}
}

\author{Weiqin Yang}
\orcid{0000-0002-5750-5515}
\email{tinysnow@zju.edu.cn}
\authornotemark[2]
\authornotemark[3]
\affiliation{%
  \institution{Zhejiang University}
  \city{Hangzhou}
  \state{Zhejiang}
  \country{China}
}

\author{Hang Pan}
\orcid{0000-0001-8020-203X}
\email{hungpaan@mail.ustc.edu.cn}
\affiliation{%
  \institution{University of Science and Technology of China}
  \city{Hefei}
  \state{Anhui}
  \country{China}
}

\author{Bingde Hu}
\orcid{0000-0003-2556-9239}
\email{tonyhu@zju.edu.cn}
\affiliation{%
  \institution{Bangsun Technology}
  \city{Hangzhou}
  \country{China}
}

\author{Can Wang}
\email{wcan@zju.edu.cn}
\orcid{0000-0002-5890-4307}
\authornotemark[2]
\authornotemark[3]
\affiliation{%
  \institution{Zhejiang University}
  \city{Hangzhou}
  \state{Zhejiang}
  \country{China}
}

\author{Jiawei Chen}
\orcid{0000-0002-4752-2629}
\email{sleepyhunt@zju.edu.cn}
\authornote{Corresponding author.}
\authornote{State Key Laboratory of Blockchain and Data Security, Zhejiang University.}
\authornote{College of Computer Science and Technology, Zhejiang University.}
\affiliation{%
  \institution{Zhejiang University}
  \city{Hangzhou}
  \state{Zhejiang}
  \country{China}
}

\renewcommand{\shortauthors}{Hongchen Li et al.}

\begin{abstract}
Large language models (LLMs) have recently emerged as powerful backbones for recommender systems by reformulating recommendation as a token-level generation task. Despite their promise, we identify a pervasive yet underexplored issue:  \textit{Length Bias}. Because items are represented by textual descriptions of varying lengths, LLM-based recommenders can be systematically biased in two ways. On the input side, longer item descriptions occupy more tokens in the context and thus receive disproportionately large aggregate attention mass during user preference modeling. On the output side, decoding based on summed autoregressive log-likelihood score inherently disfavors long items. Worse still, conventional length normalization can introduce an additional bias and even degrade recommendation performance.

To address this problem, we propose \textbf{LBR} (\textbf{L}ength \textbf{B}ias \textbf{R}eduction), a lightweight and model-agnostic framework for mitigating length bias in LLM-based recommendation. LBR mitigates input-side bias via \textit{Length-Aware Attention Calibration}, which incorporates a length-dependent offset into attention logits to neutralize attention skew.
For the output side, LBR introduces \textit{Effective Information Length Normalization}, replacing naive token count with an information-theoretic length surrogate derived from the branching structure of the prefix tree.
Extensive experiments on three real-world Amazon datasets and two representative LLM-based recommenders demonstrate that LBR substantially alleviates length bias while consistently improving recommendation accuracy and fairness, with negligible additional training and inference overhead (with an average NDCG@5 gain of $16.82\%$). The code is available at \url{https://github.com/Void-JackLee/LBR}.
\end{abstract}




\maketitle

\section{Introduction}

\begin{figure*}[t]
  \centering
  \includegraphics[width=\linewidth]{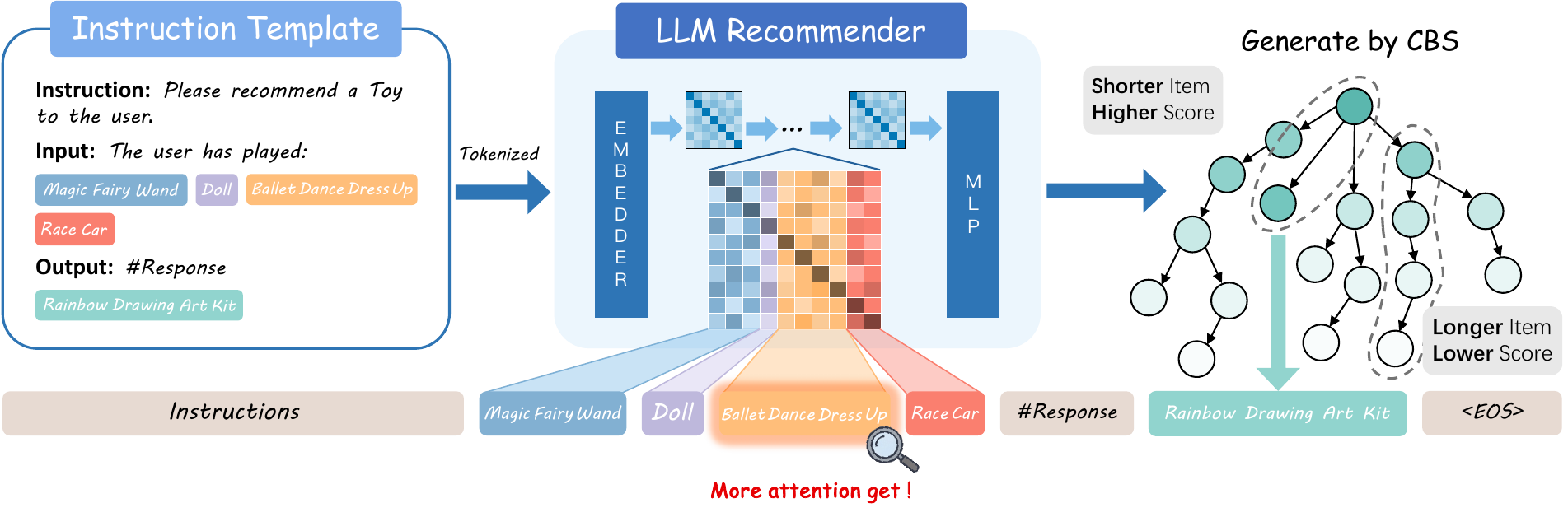}
  \caption{
  Illustration of LLM-based recommendations and length bias, wherein item length affects attention accumulation values and score computation during Trie-Constrained decoding. Longer items will get more attention, meanwhile shorter items are preferred by decoding.
  }
  \Description { long items with fewer hard tokens and shorter items with more hard tokens}
  \label{fig:attention}
\end{figure*}

\begin{figure}[t]
  \centering
  \includegraphics[width=0.9\linewidth]{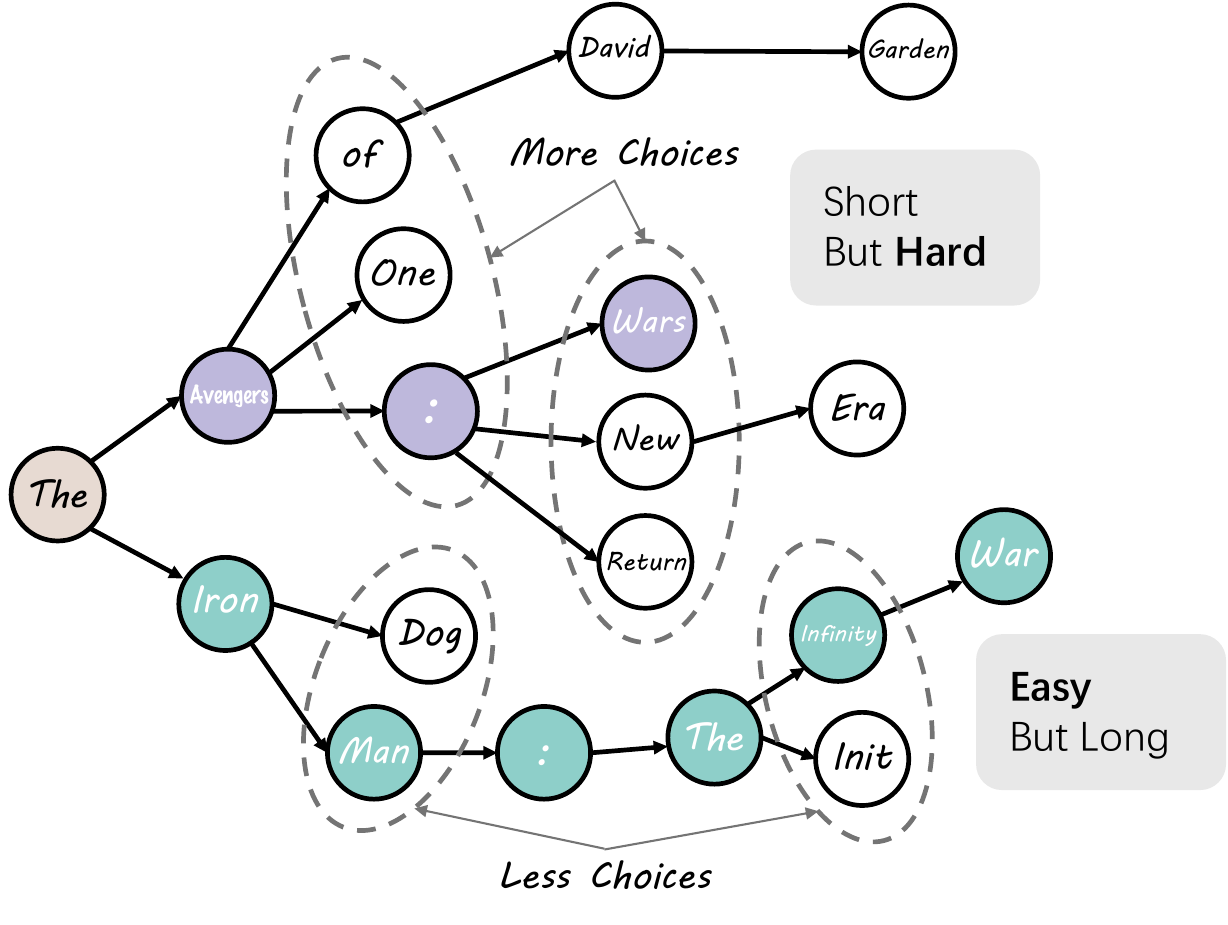}
  \caption{ Visualization of Trie-Constrained decoding. During decoding, the token generated at each step is constrained to a set of valid tokens, and the size of this valid-token set varies across steps.
  }
  \Description { long items with fewer hard tokens and shorter items with more hard tokens}
  \label{fig:decision-on-tree}
\end{figure}

Large Language Models (LLMs) have demonstrated remarkable progress across a wide range of domains \cite{multitask, squad-questions, qwen25math}, prompting a growing body of research to investigate their application in recommender systems (RS) \cite{wu2024survey}. A prominent paradigm is LLM-based recommendation \cite{wang2025msl, li2023large, bao2025bigrec, bao2023tallrec, zhu2024collaborative, wang2023zero, tan2024idgenrec, zheng2024harnessing, wang2024flip, LLaRA, kim2024large}, which directly employs an LLM as the recommendation backbone. 
As illustrated in Figure~\ref{fig:attention}, this paradigm typically reformulates recommendation as a token-level language generation task. It first represents items using textual descriptions (\eg titles), and then constructs language prompts from users’ historical interactions to instruct LLMs to predict their future preferences.
During inference, the model autoregressively generates an item description token by token, typically using Constrained Beam Search (CBS) \cite{hokamp2017lexically}, where decoding is guided by a prefix tree (\ie Trie) to strictly confine token selection to valid catalog items  (\cf Figure~\ref{fig:decision-on-tree}).

Despite these promising advances, we identify a pervasive \textit{length bias} in LLM-based recommendation: in practice, items often vary substantially in the length of their textual descriptions (\eg titles and summaries), leading to systematic bias in recommendation from two aspects:
\begin{itemize}[left=5pt]
\item \textbf{Input-Side Length Bias in Self-Attention Modeling.} LLMs are fundamentally built upon token-wise self-attention mechanisms. However, in recommendation scenarios, items with longer token sequences are naturally over-represented in the input context, and therefore tend to receive disproportionately larger aggregate attention mass during modeling (\cf Figure~\ref{fig:attention}). Consequently, longer items may exert an outsized influence on the modeling of user preference and further skew recommendation generation, irrespective of their actual relevance to the user's underlying interests.
\item \textbf{Output-Side Length Bias in Constrained Decoding.} During the decoding, items are retrieved by scoring their textual realizations using the autoregressive log-likelihood. Since log-probabilities are non-positive, longer sequences accumulate more negative terms and therefore tend to receive lower scores (\cf Figure~\ref{fig:observe_lp1}). A widely adopted remedy is \textit{length normalization}, which discounts scores by their token length. However, we find this strategy fails in  recommendation scenarios --- biasing recommendations toward longer items and degrading overall accuracy. The key reason is that tokens are not equally informative under Trie-constrained decoding: some tokens appear at high-branching Trie nodes and thus carry greater uncertainty, whereas others appear at low-branching nodes and are much easier to predict. Since longer items often contain a higher proportion of low-branching, high-probability tokens (\cf Figure~\ref{fig:decision-on-tree}), naively dividing by the total token length can artificially inflate the scores of longer items.
\end{itemize}
These findings motivate our core research question: \textbf{How can we mitigate length bias in LLM-based recommenders?}

To address these challenges, we propose \textbf{LBR} (\textbf{L}ength \textbf{B}ias \textbf{R}eduction), a lightweight and model-agnostic framework that mitigates length bias at both stages. 
For attention, LBR performs \textit{Length-Aware Attention Calibration} by introducing a length-dependent offset into the attention logits, dynamically neutralizing the attention skew incurred by heterogeneous item lengths.
For decoding, LBR introduces \textit{Effective Information Length Normalization}. Departing from naive token-count length, this module quantifies token informativeness using the branching factor at each step. Specifically, LBR quantifies token-level informativeness based on Hartley entropy \cite{aczel1974shannon} derived from its branch width, and aggregates these values as an information-theoretic length surrogate for each item. This metric explicitly accounts for differences in branching factors and can be interpreted as the number of equivalent binary decisions required to identify an item. By normalizing scores with this information-aware length, LBR enables a fairer comparison among candidates with heterogeneous textual lengths and decoding structures.

Notably, LBR is easy to implement and computationally efficient, requiring only the addition of a lightweight offset term to the attention computation and to each decoding step. Moreover, it can be seamlessly integrated into a variety of LLM-based recommender systems, improving both accuracy and fairness. Extensive experiments on three real-world datasets and two representative LLM-based recommendation backbones show that LBR substantially mitigates length bias while consistently enhancing recommendation accuracy and fairness, with negligible additional training and inference overhead.

In summary, our contributions are as follows:
\begin{itemize}[leftmargin=*]
\item We identify and analyze the length bias present in LLM-based recommender systems and show that it can negatively affect both recommendation accuracy and fairness.
\item We propose LBR, an efficient and plug-and-play debiasing method that mitigates length bias in both attention and decoding.
\item We conduct extensive experiments on three real-world datasets, demonstrating that LBR achieves state-of-the-art performance while effectively alleviating length bias.
\end{itemize}

\section{Preliminaries}
This section presents the essential background of LLM-based recommendation systems (Subsection~\ref{sec:pre_LLM4Rec}) and attention mechanisms (Subsection~\ref{sec:pre_attention}).

\subsection{LLM-based Recommendation}
\label{sec:pre_LLM4Rec}
Following prior work~\cite{LLaRA, bao2025bigrec, na2024enhancing, bao2023tallrec, lin2024rella, zheng2024adapting}, we study sequential recommendation (SR), a widely deployed setting in real-world recommender systems. Let $\mathcal{I}$ denote the item set. Given a user’s historical interaction sequence $S = \{ i_{1}, i_{2}, \dots, i_{n} \}$, where $i_t \in \mathcal I$ denotes the $t$-th interacted item, the goal is to predict the next item $i_{n+1}$ that the user is most likely to interact with. 

Motivated by the strong semantic understanding and reasoning capabilities of large language models (LLMs) \cite{dubey2024llama, achiam2023gpt, guo2025deepseek, team2024qwen2}, recent studies have explored LLMs for SR by reformulating recommendation as a language generation problem. Such methods typically comprise two components:

\textit{(1) Prompt construction}: As illustrated in Figure~\ref{fig:attention}, an LLM-based recommender first represents each item \(i_t\) using its textual description (e.g., title), \(\mathbf{x}_t=(x_{t,1},x_{t,2},\ldots,x_{t,|\mathbf{x}_t|})\), where \(x_{t,j}\) is the \(j\)-th token and \(|\mathbf{x}_t|\) is the token sequence length. It then flattens and concatenates item descriptions, optionally interleaved with auxiliary tokens (e.g., task instructions, separators), to construct a prompt:
$$
\mathbf{X}=\big[\ldots, (x_{1,1},x_{1,2},\ldots,x_{1,|\mathbf{x}_1|}), \ldots, (x_{t,1},x_{t,2},\ldots,x_{t,|\mathbf{x}_t|}), \ldots \big].
$$

\textit{(2) Item generation:} Given the prompt \(\mathbf{X}\), the LLM autoregressively generates the token sequence of the predicted item \(y=(y_1,\ldots,y_{|y|})\) according to the conditional probability $P_{\theta}(y_k | y_{< k}, \mathbf{X})$,
where $y_k$ is the $k$-th token of $y$, $y_{<k}$ denotes the prefix tokens preceding $y_k$, and $\theta$ denotes model parameters. 
During decoding, constrained beam search is often used to improve efficiency and ensure validity. Specifically, it extends partial sequences by selecting high-probability next tokens while enforcing that each prefix corresponds to a valid item in the system. This procedure can be interpreted as sequential decision-making over a prefix tree (Trie), where generation corresponds to traversing from the root to a leaf node.

\subsection{Attention Mechanism}
\label{sec:pre_attention}
Modern LLMs are typically built on the Transformer architecture, whose core operation is the self-attention mechanism \cite{vaswani2017attention}. Given an input sequence $\mathbf{X}=(x_1,\ldots,x_L)$, each token $x_u$ aggregates information from other tokens $x_v$ using attention weights $\mathbf{A}_{uv}$:
\begin{equation}
\mathbf{O}_u = \sum_{v=1}^{L} \mathbf{A}_{uv}\,\mathbf{V}_v
\end{equation}
\begin{equation}
\label{eq:attention}
\mathbf{A}_{uv} = \mathrm{Softmax}_v\!\left(\frac{\mathbf{Q}_u\mathbf{K}_v^\top}{\sqrt{d}}+\mathbf{M}_{uv}\right).
\end{equation}
Here, \(\mathbf{O}_u\) represents the output representation of token \(x_u\) after integrating information from other tokens. The query, key, and value matrices \((\mathbf{Q},\mathbf{K},\mathbf{V})\) are obtained from the input token embeddings via learned linear projections, and \(d\) denotes the embedding dimension. During autoregressive generation, the causal mask \(\mathbf{M}\) prevents the token from attending to future positions. Intuitively, self-attention captures contextual dependencies by assigning context-dependent weights to other tokens, thereby modeling token-level interactions across the sequence.

\section{Length Bias in LLM-based RS}

LLM-based recommenders model user preferences at the token level by encoding item texts and exploiting self-attention and autoregressive decoding. While this paradigm has shown promising results \cite{wang2025msl, bao2025bigrec}, it introduces a critical side effect: \textit{\textbf{Length Bias}}. In real-world platforms, items exhibit substantial variation in textual description length (e.g., titles and summaries), which can systematically bias both (i) attention-based preference modeling and (ii) decoding-based item scoring. Such biases may degrade recommendation accuracy and raise fairness concerns. We detail these two effects below.

\subsection{Length Bias in Attention}
\label{subsec:analyze_attention}

LLM-based recommenders represent each historical item $i_t$ as a token sequence $\mathbf{x}_t$ with length $|\mathbf{x}_t|$, and then model fine-grained token interactions via self-attention to infer user preference. However, the lengths $|\mathbf{x}_t|$ can vary widely across items  --- an ubiquitous and often unavoidable phenomenon, since different items naturally require different amounts of text to describe their semantics. This variation can induce a \textit{length advantage} in attention: items with longer sequences contribute more tokens to the input and therefore provide more opportunities to receive attention. Consequently, an item may exert a larger influence on the model prediction simply due to having more tokens, rather than being more relevant to the user’s preference. 

To verify this intuition, we quantify each item's attention effect by accumulating the attention weights assigned to its tokens when predicting the next item. 
Figure~\ref{fig:attention_length_bias_analysis} shows a clear positive correlation between item length and cumulative attention, indicating that longer items tend to exert disproportionately larger influence on predictions. 
We also empirically verified this via a feature ablation study \cite{miglani2023using} measuring each historical item's cumulative token effect on the generated recommendation (detailed in Appendix~\ref{appsec:captum_attr_analysis}).
This passive amplification can distort preference inference and skew recommendation outcomes.

\begin{figure}[t]
  \centering
  \includegraphics[width=\linewidth]{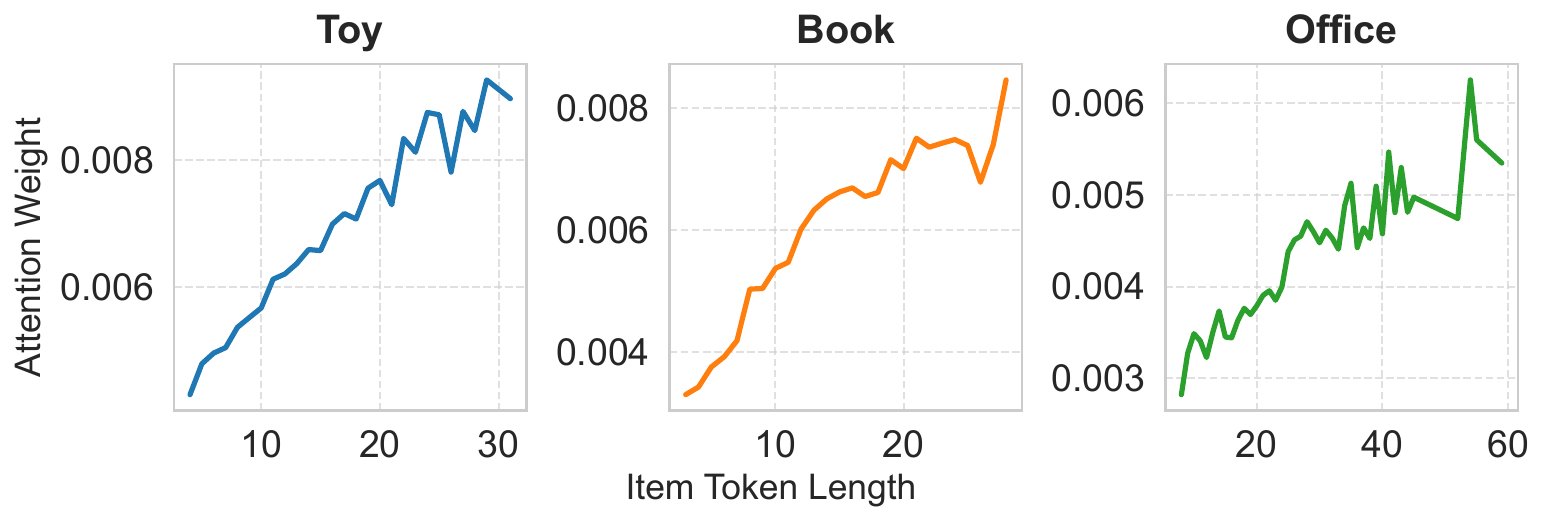}
\caption{ Illustration of the relationship between an item’s cumulative attention weight in prediction and its token length.}  
  
  \Description{ The correlation between item token length and attention weights. }
  \label{fig:attention_length_bias_analysis}
\end{figure}

\textbf{Limitations of naive strategies.} A straightforward mitigation is to enforce equal-length item representations, for example via padding/truncation or by summarizing item descriptions to a fixed token budget. However, these operations may alter or remove semantically critical content and disrupt the original textual structure, leading to degraded recommendation quality. Our experiments in Appendix~\ref{appsec:result_naive} confirm that such preprocessing yields no consistent gains and can even hurt performance. Another potential solutions are to tokenize item text into discrete semantic IDs (e.g., via quantization methods such as RQ-VAE~\cite{rajput2023recommender}). While this reduces length variation, it weakens the advantages of LLMs: discrete codes are less interpretable, may hinder direct semantic understanding in natural language, and typically require substantial additional training data to learn the mapping between codes and semantics. We have also test such strategies in the standard LLM-based recommendation experimental setting and observe suboptimal performance (\cf Appendix~\ref{appsec:result_naive}).

\subsection{Length Bias in Decoding}
\label{sec:analyze_decoding}
In LLM-based recommendation, items are generated by scoring their textual realizations using the autoregressive log-likelihood with:
\begin{equation}
\label{eq:score_with_no_lp}
s(y) = \sum_{k = 1}^{|y|} \log P_{\theta}(y_k | y_{< k}, \mathbf{X})
\end{equation}
This scoring function inherently favors shorter items. Because log-probabilities are non-positive, longer sequences accumulate more negative terms and therefore tend to receive lower scores (cf. Figure~\ref{fig:observe_lp0}). This length bias can affect not only accuracy but also fairness and broader ecosystem incentives (e.g., by encouraging providers to adopt unnaturally short titles).
We also examine item interaction frequencies across length groups to verify that popularity does not affect our findings (detailed in Appendix~\ref{appsec:popularity_analysis}).

\textbf{Limitations of length normalization.} A common remedy is introducing length normalization (\aka length penalty):
\begin{equation}
s(y) = \frac{1}{|y|^\alpha} \sum_{k = 1}^{|y|} \log P_{\theta}(y_k | y_{< k}, \mathbf{X})
\end{equation}
where \(\alpha\) is a hyperparameter that controls the strength of the penalty.  
Intuitively, this transformation converts the cumulative score into an average-like score. However, we find that this heuristic can bias recommendations toward longer items (\cf Figure~\ref{fig:observe_lp1}) and may even degrade recommendation accuracy even when $\alpha$ approaching $0$ (\cf Appendix~\ref{appsec:result_tuning_alpha}).

We attribute this failure to a key oversight: \textit{length normalization treats all tokens as equally informative}. In LLM-based recommendation, decoding is often constrained by a prefix tree (Trie) over valid item tokens. Under such constraints, token prediction difficulty varies substantially across Trie nodes: some positions have many feasible tokens (high branching factor) and are inherently uncertain, while others have only a few options and are easy to predict with high probability. Empirically, tokens at low-branching nodes often receive higher probabilities, whereas tokens at high-branching nodes receive much lower probabilities (\cf Figure~\ref{fig:avg_probs_k}). Longer item often contain a higher fraction of ``easy'' tokens. Therefore, dividing by $|y|^\alpha$ without accounting for token informativeness can unintentionally inflate the scores of long items.

\textbf{Limitations of D3 \cite{bao2024decoding}.}  
A related work D3 also attempts to mitigate the issue that length normalization tends to amplify the scores of items containing many tokens with probabilities close to 1 (``ghost tokens''). D3 empirically observes that, after removing such ghost tokens, the effective lengths of all items become nearly identical. Based on this observation, D3 adopts a simple strategy that completely removes length normalization by setting $\alpha$ = 0.
However, this overly coarse approach fails to effectively alleviate length bias. As reflected in our experiments and also reported in D3, the best recommendation performance is often achieved when $\alpha$ is close to 0, under which substantial length bias still remains (\cf Figure~\ref{fig:observe_lp0}). Moreover, D3 still does not explicitly distinguish or quantify the informativeness of tokens at multi-branch nodes, despite the fact that such tokens are prevalent and play a crucial role in score computation.

\begin{figure}[t]
  \centering
  \includegraphics[width=\linewidth]{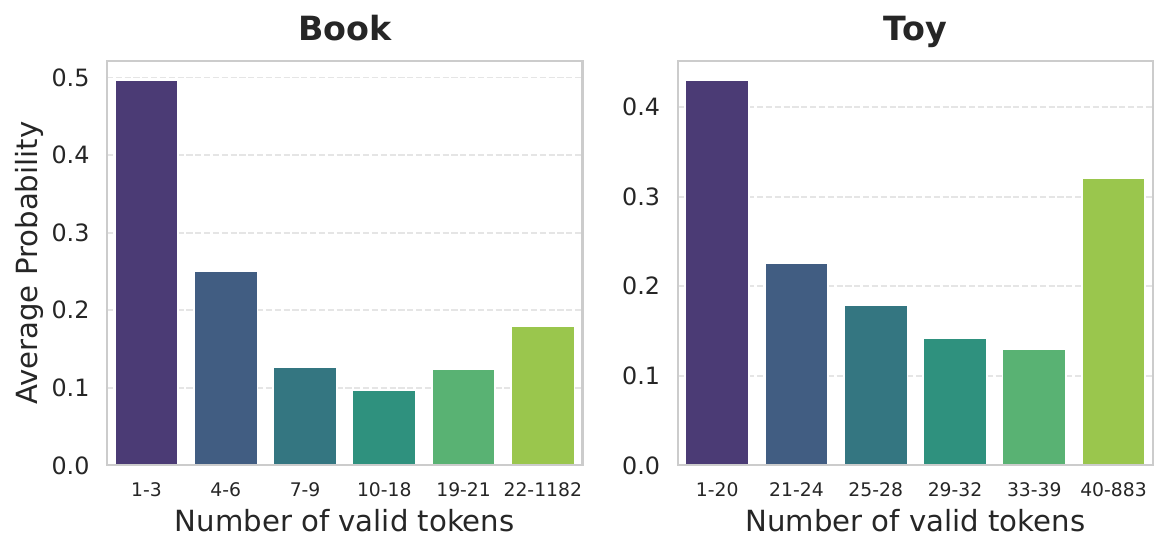}
  \caption{ 
    Average probability distribution across number of valid tokens.
    Tokens at low-branching nodes often receive high probabilities.}
  \label{fig:avg_probs_k}
\end{figure}






\begin{figure*}[t]
  \centering

  \begin{subfigure}{0.48\textwidth}
    \centering
    \includegraphics[width=\linewidth]{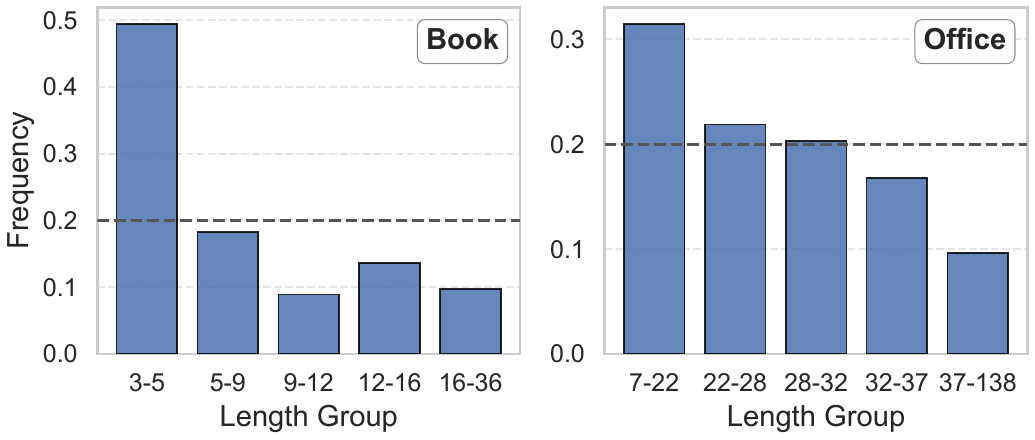}
    \caption{Without length penalty}
    \label{fig:observe_lp0}
  \end{subfigure}
  \hfill 
  \begin{subfigure}{0.48\textwidth}
    \centering
    \includegraphics[width=\linewidth]{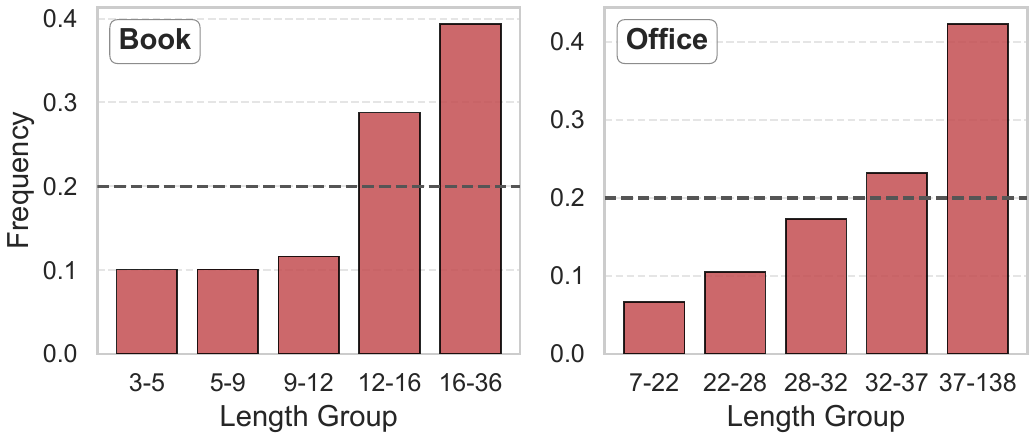}
    \caption{With standard length penalty}
    \label{fig:observe_lp1}
  \end{subfigure}
  \vspace{-0.2cm} 
  \caption{ 
    Illustration of length bias in decoding. Items are partitioned into five groups by token length, with each group containing the same cumulative item frequency in the test set. The dashed line denotes the ideal exposure distribution of each group implied by the test data. (a) Without length penalty, the model is strongly biased toward shorter items. (b) Standard length normalization shifts the bias toward longer items.
  }
  \Description{ Analysis of Length Bias on decoding. }
\end{figure*}

\section{Methodology}

In this section, we present \textbf{LBR} (\textbf{L}ength \textbf{B}ias \textbf{R}eduction), a novel framework designed to mitigate length bias in both the attention-based preference modeling and the decoding-based item scoring.

\subsection{Length-Aware Attention Calibration}
\label{subsec:method_attn_bias}

Section~\ref{subsec:analyze_attention} shows that, in LLM-based SR, an item's cumulative attention weight tends to increase with the length of its textual description, producing an undesirable positive correlation between attention influence and token length. For convenience, we denote such inherent bias as $g(l)$, a mapping function from item length $l$ to the expected cumulative attention mass (\ie the length–attention trend). 

To counteract this bias, we introduce a length-dependent offset term into the attention logits, which downweights tokens from long items and upweights tokens from short items. Formally, we modify the attention weight from Eq.\ref{eq:attention} into:
\begin{equation}
\mathbf{A}'_{uv}=\mathrm{Softmax}_v\!\left(\frac{\mathbf{Q}_u \mathbf{K}_v^\top}{\sqrt{d}}+\mathbf{M}_{uv}+{\color{red}\delta(l_v)}\right).
\end{equation}
where $l_v$ denotes the length of the item to which token $v$ belongs. For those auxiliary tokens, we simply set $l_v$ as the average lengths of the items on the historical interactions. We further set:
\begin{equation}
    \delta(l_v) = -\log(g(l_v))
\end{equation}

Intuitively, because softmax applies an exponential transformation to logits, subtracting $\log(g(l_v))$ in the logit space approximately cancels the multiplicative ``length advantage'' induced by $g(l_v)$,  thereby reducing the correlation between cumulative attention mass and item length. In Appendix~\ref{appsec:laac_invariance}, we provide an theoretical analysis showing that, under mild assumptions, the expected cumulative attention weights becomes approximately length-invariant after calibration (\cf Lemma~\ref{lem:laac_invariance}). Empirically, we observe that this calibration substantially reduces attention length bias (\cf Section~\ref{sec:exp_depth_analyze}).


In practice, we model $g(l)$ as a simple linear function, \ie $g(l)=a l+b$. The reasons are multiple folds: 
(1) it well approximates the linear trend observed in Figure~\ref{fig:attention_length_bias_analysis}
and
(2) it ensures ease of implementation and optimization. Although more complex non-linear models possess higher expressivity, we found they provide negligible empirical gains over the linear counterpart.

Now the question lies in how to learn this function. We could estimate it directly from historical training data as depicted in Figure~\ref{fig:attention_length_bias_analysis}, but this inevitably requires an additional training stage. In practice, we jointly optimize the parameters of this function and the recommendation model through end-to-end training, guided by the primary recommendation objective. We find that to achieve better recommendation performance, the model adaptively learns an effective $g(l)$ to alleviate the bias issue.


\useunder{\uline}{\ul}{}
\begin{table*}
\caption{ The performance comparison on three real-world datasets.
The LLM-based RS methods are applied ("+") on the BIGRec and LLaRA backbone.
The best result is bolded. "{\color[HTML]{DA5F00} Imp.\%}" denotes the improvement of LBR over the best baseline. "N" represents NDCG, and "H" represents Hit Ratio. The improvements are statistically significant (p<0.05).}
\label{tab:cmp_exp}
\begin{adjustbox}{width=\textwidth}
\begin{tabular}{@{}lcccccccccccc@{}}
\toprule
\multicolumn{1}{c}{}                                             & \multicolumn{4}{c}{\textbf{Toy}}                                                                                                                                                          & \multicolumn{4}{c}{\textbf{Office}}                                                                                                                                                        & \multicolumn{4}{c}{\textbf{Book}}                                                                                                                                   \\ \cmidrule(l){2-13} 
\multicolumn{1}{c}{\multirow{-2}{*}{\textbf{Methods}}}           & \textbf{N@5}                            & \textbf{N@10}                           & \textbf{H@5}                            & \multicolumn{1}{c|}{\textbf{H@10}}                          & \textbf{N@5}                            & \textbf{N@10}                           & \textbf{H@5}                            & \multicolumn{1}{c|}{\textbf{H@10}}                           & \textbf{N@5}                            & \textbf{N@10}                          & \textbf{H@5}                            & \textbf{H@10}                          \\ \midrule
\multicolumn{1}{l|}{SASRec}                                      & 0.0089                                  & 0.0109                                  & 0.0148                                  & \multicolumn{1}{c|}{0.0211}                                 & 0.0127                                  & 0.0196                                  & 0.0223                                  & \multicolumn{1}{c|}{0.0447}                                  & 0.0061                                  & 0.0080                                 & 0.0100                                  & 0.0157                                 \\
\multicolumn{1}{l|}{DROS}                                        & 0.0105                                  & 0.0126                                  & 0.0148                                  & \multicolumn{1}{c|}{0.0213}                                 & 0.0195                                  & 0.0240                                  & 0.0265                                  & \multicolumn{1}{c|}{0.0405}                                  & 0.0072                                  & 0.0096                                 & 0.0116                                  & 0.0192                                 \\ \midrule
\multicolumn{1}{l|}{LLM-CF}                                      & 0.0077                                  & 0.0101                                  & 0.0148                                  & \multicolumn{1}{c|}{0.0223}                                 & 0.0161                                  & 0.0200                                  & 0.0306                                  & \multicolumn{1}{c|}{0.0426}                                  & 0.0077                                  & 0.0100                                 & {\ul 0.0141}                            & {\ul 0.0212}                           \\
\multicolumn{1}{l|}{DLLM2Rec}                                    & 0.0080                                  & 0.0110                                  & 0.0156                                  & \multicolumn{1}{c|}{0.0248}                                 & 0.0166                                  & 0.0231                                  & 0.0322                                  & \multicolumn{1}{c|}{0.0530}                                  & 0.0068                                  & 0.0094                                 & 0.0112                                  & 0.0194                                 \\ \midrule
\multicolumn{1}{l|}{BIGRec}                                      & 0.0123                                  & 0.0155                                  & 0.0209                                  & \multicolumn{1}{c|}{0.0305}                                 & 0.0105                                  & 0.0226                                  & 0.0213                                  & \multicolumn{1}{c|}{0.0587}                                  & 0.0062                                  & 0.0090                                 & 0.0116                                  & 0.0201                                 \\
\multicolumn{1}{l|}{+D3}                                         & 0.0130                                  & 0.0162                                  & 0.0204                                  & \multicolumn{1}{c|}{0.0302}                                 & 0.0142                                  & 0.0237                                  & 0.0291                                  & \multicolumn{1}{c|}{0.0592}                                  & 0.0058                                  & 0.0081                                 & 0.0110                                  & 0.0180                                 \\
\multicolumn{1}{l|}{+CFT}                                        & 0.0138                                  & 0.0173                                  & 0.0217                                  & \multicolumn{1}{c|}{0.0325}                                 & 0.0183                                  & 0.0272                                  & 0.0306                                  & \multicolumn{1}{c|}{0.0582}                                  & 0.0060                                  & 0.0092                                 & 0.0103                                  & 0.0205                                 \\
\multicolumn{1}{l|}{+S-DPO}                                      & 0.0181                                  & 0.0227                                  & {\ul 0.0288}                            & \multicolumn{1}{c|}{{\ul 0.0432}}                           & {\ul 0.0209}                            & {\ul 0.0292}                            & {\ul 0.0374}                            & \multicolumn{1}{c|}{{\ul 0.0629}}                            & {\ul 0.0082}                            & {\ul 0.0102}                           & 0.0125                                  & 0.0189                                 \\
\multicolumn{1}{l|}{+Reweight}                                   & 0.0084                                  & 0.0121                                  & 0.0146                                  & \multicolumn{1}{c|}{0.0261}                                 & 0.0053                                  & 0.0122                                  & 0.0125                                  & \multicolumn{1}{c|}{0.0338}                                  & 0.0051                                  & 0.0070                                 & 0.0100                                  & 0.0160                                 \\
\multicolumn{1}{l|}{+IGD}                                        & {\color[HTML]{000000} {\ul 0.0191}}     & {\color[HTML]{000000} {\ul 0.0239}}     & {\color[HTML]{000000} 0.0273}           & \multicolumn{1}{c|}{{\color[HTML]{000000} 0.0421}}          & {\color[HTML]{000000} 0.0187}           & {\color[HTML]{000000} 0.0266}           & {\color[HTML]{000000} 0.0312}           & \multicolumn{1}{c|}{{\color[HTML]{000000} 0.0566}}           & {\color[HTML]{000000} 0.0064}           & {\color[HTML]{000000} 0.0084}          & {\color[HTML]{000000} 0.0110}           & {\color[HTML]{000000} 0.0171}          \\
\multicolumn{1}{l|}{{\color[HTML]{000000} \textbf{+LBR}}}        & {\color[HTML]{000000} \textbf{0.0195}}  & {\color[HTML]{000000} \textbf{0.0241}}  & {\color[HTML]{000000} \textbf{0.0315}}  & \multicolumn{1}{c|}{{\color[HTML]{000000} \textbf{0.0459}}} & {\color[HTML]{000000} \textbf{0.0224}}  & {\color[HTML]{000000} \textbf{0.0321}}  & {\color[HTML]{000000} \textbf{0.0457}}  & \multicolumn{1}{c|}{{\color[HTML]{000000} \textbf{0.0758}}}  & {\color[HTML]{000000} \textbf{0.0098}}  & {\color[HTML]{000000} \textbf{0.0112}} & {\color[HTML]{000000} \textbf{0.0167}}  & {\color[HTML]{000000} \textbf{0.0213}} \\ \midrule
\multicolumn{1}{l|}{{\color[HTML]{DA5F00} \textbf{Improvement}}} & {\color[HTML]{DA5F00} \textbf{2.09\%}}  & {\color[HTML]{DA5F00} \textbf{0.84\%}}  & {\color[HTML]{DA5F00} \textbf{9.38\%}}  & \multicolumn{1}{c|}{{\color[HTML]{DA5F00} \textbf{6.25\%}}} & {\color[HTML]{DA5F00} \textbf{7.18\%}}  & {\color[HTML]{DA5F00} \textbf{9.93\%}}  & {\color[HTML]{DA5F00} \textbf{22.19\%}} & \multicolumn{1}{c|}{{\color[HTML]{DA5F00} \textbf{20.51\%}}} & {\color[HTML]{DA5F00} \textbf{19.51\%}} & {\color[HTML]{DA5F00} \textbf{9.80\%}} & {\color[HTML]{DA5F00} \textbf{18.44\%}} & {\color[HTML]{DA5F00} \textbf{0.47\%}} \\ \midrule
\multicolumn{1}{l|}{LLaRA}                                       & {\color[HTML]{000000} 0.0136}           & {\color[HTML]{000000} 0.0169}           & {\color[HTML]{000000} 0.0213}           & \multicolumn{1}{c|}{{\color[HTML]{000000} 0.0313}}          & {\color[HTML]{000000} 0.0108}           & {\color[HTML]{000000} 0.0236}           & {\color[HTML]{000000} 0.0218}           & \multicolumn{1}{c|}{{\color[HTML]{000000} 0.0613}}           & {\color[HTML]{000000} 0.0073}           & {\color[HTML]{000000} 0.0096}          & {\color[HTML]{000000} 0.0128}           & {\color[HTML]{000000} 0.0201}          \\
\multicolumn{1}{l|}{+D3}                                         & 0.0137                                  & 0.0172                                  & 0.0209                                  & \multicolumn{1}{c|}{0.0315}                                 & 0.0116                                  & 0.0240                                  & 0.0234                                  & \multicolumn{1}{c|}{{\ul 0.0618}}                            & 0.0066                                  & 0.0094                                 & 0.0112                                  & 0.0194                                 \\
\multicolumn{1}{l|}{+CFT}                                        & 0.0151                                  & 0.0179                                  & 0.0232                                  & \multicolumn{1}{c|}{0.0317}                                 & {\ul 0.0168}                            & {\ul 0.0254}                            & {\ul 0.0337}                            & \multicolumn{1}{c|}{0.0597}                                  & {\ul 0.0077}                            & {\ul 0.0102}                           & {\ul 0.0144}                            & {\ul 0.0210}                           \\
\multicolumn{1}{l|}{+S-DPO}                                      & {\ul 0.0167}                            & {\ul 0.0219}                            & {\ul 0.0286}                            & \multicolumn{1}{c|}{{\ul 0.0446}}                           & 0.0101                                  & 0.0216                                  & 0.0213                                  & \multicolumn{1}{c|}{0.0577}                                  & 0.0050                                  & 0.0068                                 & 0.0098                                  & 0.0153                                 \\
\multicolumn{1}{l|}{+Reweight}                                   & 0.0094                                  & 0.0133                                  & 0.0163                                  & \multicolumn{1}{c|}{0.0286}                                 & 0.0106                                  & 0.0185                                  & 0.0234                                  & \multicolumn{1}{c|}{0.0478}                                  & 0.0066                                  & 0.0083                                 & 0.0121                                  & 0.0173                                 \\
\multicolumn{1}{l|}{+IGD}                                        & 0.0137                                  & 0.0174                                  & 0.0202                                  & \multicolumn{1}{c|}{0.0317}                                 & 0.0099                                  & 0.0212                                  & 0.0229                                  & \multicolumn{1}{c|}{0.0582}                                  & 0.0068                                  & 0.0097                                 & 0.0114                                  & 0.0203                                 \\
\multicolumn{1}{l|}{\textbf{+LBR}}                               & \textbf{0.0208}                         & \textbf{0.0253}                         & \textbf{0.0334}                         & \multicolumn{1}{c|}{\textbf{0.0474}}                        & \textbf{0.0213}                         & \textbf{0.0300}                         & \textbf{0.0353}                         & \multicolumn{1}{c|}{\textbf{0.0623}}                         & \textbf{0.0093}                         & \textbf{0.0111}                        & \textbf{0.0157}                         & \textbf{0.0212}                        \\ \midrule
\multicolumn{1}{l|}{{\color[HTML]{DA5F00} \textbf{Improvement}}} & {\color[HTML]{DA5F00} \textbf{24.55\%}} & {\color[HTML]{DA5F00} \textbf{15.53\%}} & {\color[HTML]{DA5F00} \textbf{16.78\%}} & \multicolumn{1}{c|}{{\color[HTML]{DA5F00} \textbf{6.28\%}}} & {\color[HTML]{DA5F00} \textbf{26.79\%}} & {\color[HTML]{DA5F00} \textbf{18.11\%}} & {\color[HTML]{DA5F00} \textbf{4.75\%}}  & \multicolumn{1}{c|}{{\color[HTML]{DA5F00} \textbf{0.81\%}}}  & {\color[HTML]{DA5F00} \textbf{20.78\%}} & {\color[HTML]{DA5F00} \textbf{8.82\%}} & {\color[HTML]{DA5F00} \textbf{9.03\%}}  & {\color[HTML]{DA5F00} \textbf{0.95\%}} \\ \bottomrule
\end{tabular}
\end{adjustbox}
\end{table*}

\begin{table}[t]
  \caption{Dataset statistics}
  \label{table:dataset_statistics}
\begin{tabular}{c|cccc}
\toprule
\textbf{Dataset} & \textbf{\#Users} & \textbf{\#Items} & \textbf{\#Interactions} & \textbf{Density} \\
\midrule
\textbf{Toy}     & 19124 & 11758 & 165247 & 0.0735\% \\
\textbf{Office}  & 4895            & 2414            & 53149                  & 0.4498\%         \\
\textbf{Book}    & 16559 & 6344 & 151928 & 0.1446\% \\
\bottomrule
\end{tabular}
\end{table}

\subsection{Decoding with Effective Information Length}

Section~\ref{sec:analyze_decoding} shows that decoding based on cumulative log-likelihood systematically favors short items, whereas conventional length normalization conversely skew recommendation toward long items and even degrade recommendation performance. We attribute this failure to a key oversight: tokens differ in informativeness under Trie-Constrained decoding. Some positions admit many valid tokens (high branching factor), including high uncertainty and high information content, while others are nearly deterministic (low branching factor) and carry little information. Normalizing by the raw token count ignores this heterogeneity and can disproportionately benefit items (\eg longer ones) that contain a high proportion of  low-branch tokens.

To tackle this issue, we first define a token information score using Hartley entropy \cite{aczel1974shannon} to quantify the token-level informativeness:
\begin{equation}
\label{eq:entropy}
\mu_k = \log_2 |\mathcal{V}_k| .
\end{equation}
where $\mathcal{V}_k$ denotes the set of valid tokens at step $k$ under the Trie constraint. This choice is motivated by information theory: the maximum information content of a discrete decision is bounded by the size of its feasible decision set.

We then replace physical length with \textit{Effective Information Length (EIL)}. For a candidate item sequence \(y=(y_1,y_2,\ldots,y_{|y|})\), we define:
\begin{equation}
U_y = \sum_{k=1}^{|y|} \mu_k 
    = \sum_{k=1}^{|y|} \log_2 |\mathcal{V}_k| .
\end{equation}
which measures the total uncertainty (in bits) that must be resolved along the Trie path to uniquely identify $y$. Using $U_y$ as a normalization factor yields a score that reflects confidence \textbf{per unit of information}, rather than per token. 

Finally, we compute an information-normalized score as:
\begin{equation}
s_{\mathrm{EIF}}(y) 
= \frac{1}{U_y}\sum_{k=1}^{|y|}\mu_k \cdot \log P_\theta\!\left(y_k \mid y_{<k}, X\right).
\end{equation}
which can be interpreted as the average log-likelihood per bit of information content. Compared with the standard cumulative log-likelihood, this formulation (i) normalizes by effective information length $U_y$ rather than $|y|$, thereby accounting for token-level differences in information content and prediction difficulty; and (ii)  weights each token’s log-probability by its informativeness $\mu_k$.  This design matches our intuition: a highly branched node (large valid tokens set) corresponds to a more informative decision in and should be more important, whereas near-deterministic nodes are less informative and should contribute less.

\textbf{Intuition via a binary-decision view.} An equivalent intuition is to transform the original Trie into an implicit binary decision tree: choosing among $H$ branches requires approximately $\log_2 H$  binary decisions.  Accordingly, a Trie node with branching factor $H$ can be viewed as expanding into an equivalent subtree of depth $\log_2 H$ in the binary-decision tree. Under this view, Eq.~\ref{eq:entropy} measures the per-step ``binary decision depth”, and $U_y$  becomes the effective decision length of item $|y|$ measured in uniform binary choices (\cf Appendix~\ref{appsec:bin-tree}).


\subsection{Practical Advantages}

\textbf{Easy to implement and model-agnostic.} LBR only modifies (i) attention logits via an additive offset and (ii) decoding scores via information-based weights and normalization, requiring minimal code changes. It can be integrated into a broad range of LLM-based backbones. In our experiments, LBR consistently improves representative backbones (e.g., BIGRec \cite{bao2025bigrec} and LLaRA \cite{LLaRA}) (\cf Section~\ref{sec:cmp_exp}).

\textbf{Efficiency.} The calculation of additive offset $\delta(l_v)$ and information score $\mu_k$ is highly efficient with only $O(N)$ complexity, where $N$ denotes the total number of tokens in both the input and the generated output. Empirical results (\cf Section~\ref{sec:exp_efficiency}) further show that LBR introduces almost no additional cost, demonstrating its practical efficiency.
\section{Experiments}
We aim to answer the following research questions:
\begin{itemize}[left=5pt]
\item {RQ1: How does LBR perform compared to SOTA methods?} 
\item {RQ2: How do different components of LBR affect?} 
\item {RQ3: Does LBR mitigate length bias?}
\item {RQ4: How does LBR perform compared with baselines in terms of efficiency?}
\end{itemize}

\subsection{Experimental Settings}
\subsubsection{Datasets}
Three widely used real-world datasets—\textit{Amazon Toys and Games}, \textit{Amazon Office Products}, and \textit{Amazon Books}\footnote{\url{https://jmcauley.ucsd.edu/data/amazon/index_2014.html}}—are employed in our experiments. 
These datasets have been extensively utilized in prior research on LLM-based RS \cite{cui2024distillation, bao2025bigrec, cao2024aligning, li2023e4srec, wang2025msl, lee2024star}. 
To guarantee a fair comparison with existing methods, we adopt the same data preprocessing protocols as those reported in recent studies \cite{bao2025bigrec, cui2024distillation, wang2025msl}. 
Specifically, we first apply the 5-core setting to the original datasets. User interaction sequences containing more than 11 interactions are divided using a sliding window with a fixed length of 11. The resulting subsequences are ordered chronologically and subsequently split into training, validation, and test sets according to an 8:1:1 ratio. 
Due to the large size of \textit{Amazon Books}, we randomly sample 100,000 items prior to 5-core processing. 
The statistics of the processed datasets are summarized in Table \ref{table:dataset_statistics}.


\begin{table}[t]
  \caption {
    Ablation study of LBR on NDCG@5. “w/o LAAC”, “w/o EILN”, and “w/o Both” indicate removing Length-Aware Attention Calibration, Effective Information Length Normalization, and both components, respectively.
  }
  \label{table:ablation}
  \begin{tabular}{c|cccc}
  \toprule
  \textbf{} & \textbf{Toy} & \textbf{Office} & \textbf{Book} \\
  \midrule
  LBR & 0.0195 & 0.0224 & 0.0098 \\
  w/o LAAC & 0.0192 & 0.0218 & 0.0093 \\
  w/o EILN & 0.0130 & 0.0124 & 0.0073 \\
  w/o Both & 0.0123 & 0.0105 & 0.0062 \\
  \bottomrule
  \end{tabular}
\end{table}

\subsubsection{Baselines}
The methods compared fall into several categories:
\begin{itemize}[left=5pt]
  \item \textbf{Traditional RS}: Traditional SR methods, including \underline{SASRec} (ICDM ’18) \cite{kang2018self} and \underline{DROS} (SIGIR ’23) \cite{yang2023generic}.
  \item \textbf{LLM-enhanced RS}: Conventional RS models augmented with LLMs as semantic encoders or knowledge enhancers, including \underline{DLLM2Rec} (RecSys ’24) \cite{cui2024distillation} and \underline{LLM-CF} (CIKM ’24) \cite{sun2024large}.
  \item \textbf{LLM-based RS}: Methods that directly employ LLMs for recommendation, including \underline{BIGRec} (TORS ’25) \cite{bao2025bigrec}, \underline{CFT} (CoRR ’24) \cite{zhang2024causality}, \underline{LLaRA} (SIGIR ’24) \cite{LLaRA}, \underline{S-DPO} (NeurIPS ’24) \cite{chen2024softmax}, and \underline{IGD} (NeurIPS ’25) \cite{lin2025igd}.
  \item \textbf{Debiasing methods for LLM-based RS}: Methods that mitigate bias in LLM-based recommendation systems, including \underline{Reweight} (WWW ’24) \cite{jiang2024item} and \underline{D3} (EMNLP ’24) \cite{bao2024decoding}.
\end{itemize}

\subsubsection{Implementation Details}
For all methods based on LLMs, we employ LLaMA3.2-3B \cite{dubey2024llama} as the architecture, and the training epoch is set as 5.
During inference, Constrained Beam Search (CBS) is uniformly applied to all baseline methods, following prior work \cite{wang2025msl, bao2024decoding}, to guarantee that the generated recommendations are restricted to the predefined item set. The beam width is set to 10.
For the evaluation, recommendation accuracy is assessed using two commonly adopted metrics, namely \textit{NDCG@K} and \textit{Hit Ratio@K}, with \textit{K} set to 5 and 10. We compute the metrics for the model at each checkpoint and report the final results from the epoch that achieves the highest NDCG@5. The corresponding test-set results are reported as the final evaluation outcomes.
To ensure a fair comparison, we adopt the official implementations released by the original authors and adjust the hyperparameters of all baseline models in accordance with the configurations recommended in their respective papers.

\subsection{Performance Comparison (RQ1)}
\label{sec:cmp_exp}
Table~\ref{tab:cmp_exp} presents a comparative analysis of the proposed LBR method against the baselines. LBR exhibits substantial performance gains across all datasets and consistently achieves the best results among all compared methods, with an average improvement of $16.82\%$ in NDCG@5 compared with the best baseline method.
These improvements can be attributed to two key factors: (i) by mitigating the impact of length bias on attention, LBR enables the model to correctly focus on the relevant items in the user sequence; and (ii) by alleviating the influence of length bias during decoding, LBR produces an output distribution that better aligns with the real data distribution.
Moreover, since our proposed LBR is applicable to any Transformer-based LLM, it can be seemingly integrated into other LLM-based recommendation approaches. Experiments on two LLM-based recommenders, BIGRec and LLaRA, show that LBR consistently and effectively improves their performance, demonstrating strong generalization capability.
In contrast, the other baseline methods yield only limited improvements.

\begin{figure}[t]
  \centering
  \includegraphics[width=\linewidth]{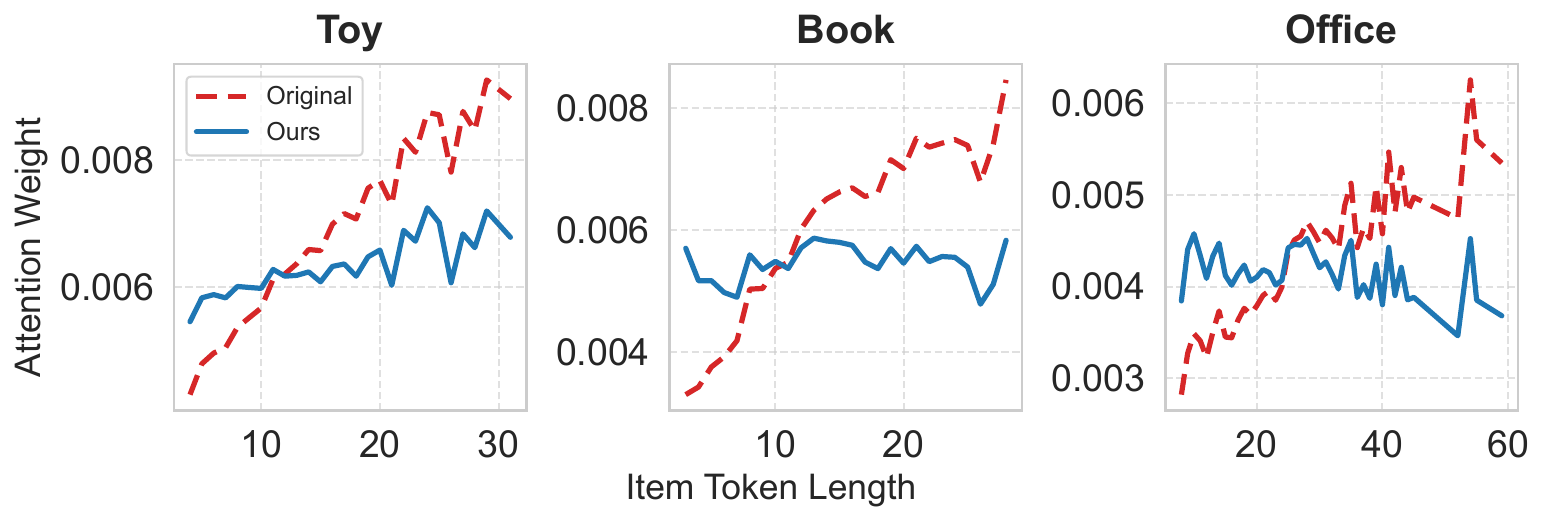}
  \caption{
  Visualization of attention weights for the baseline model and LBR across varying item lengths.
  }
  \Description{ Case Study on Attention. }
  \label{fig:attention_length_bias_comparison}
\end{figure}

\begin{figure}[t]
  \centering
  \includegraphics[width=\linewidth]{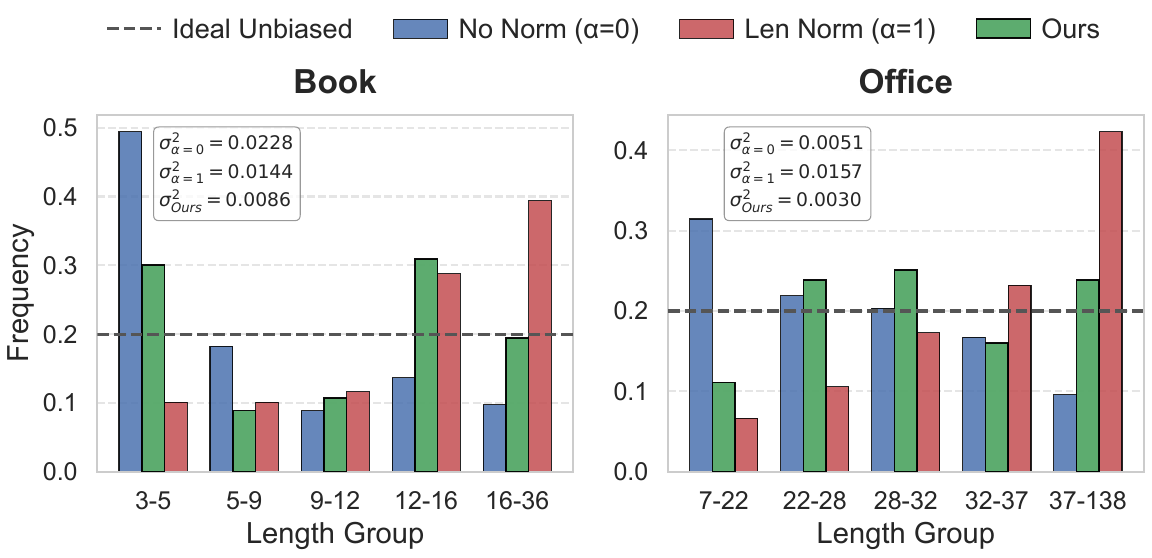}
  \caption{
  Comparison of recommendation distributions across item-length groups under different decoding strategies. Items are partitioned into five groups by token length, with each group containing the same cumulative item frequency in the test set. The variance of the recommendation proportions across groups is reported.
  }
  \Description{ Case Study on Decoding. }
  \label{fig:case_study_decoding}
\end{figure}

\subsection{Ablation Study (RQ2)}
In this section, we conduct an ablation study to investigate the contributions of different modules in LBR, including LAAC and EILN.
We compare LBR with two variants: \underline{“w/o LAAC”}, which indicates removing the LAAC module, and \underline{“w/o EILN”}, which indicates removing the EILN module. The results in Table \ref{table:ablation} show that all modules in LBR are essential.
Specifically, LAAC mitigates the tendency of the model’s attention to be dominated by longer items, enabling the model to correctly focus on important items in the sequence. EILN alleviates the impact of item length on the model’s output distribution, thereby effectively aligning the predicted distribution with the ground-truth data distribution. Together, LAAC and EILN reduce length-induced bias at different stages of the model, resulting in consistent performance improvements.

\subsection{In-depth Analysis (RQ3)}
\label{sec:exp_depth_analyze}
In this section, we conduct an empirical analysis to verify whether LBR effectively mitigates length bias. We examine the issue from two perspectives: (1) the influence of length on the attention and (2) its influence on the decoding stage.

First, LBR effectively reduces the influence of item length on attention weights, preventing longer items from systematically receiving higher attention. As shown in Fig.~\ref{fig:attention_length_bias_comparison}, attention weights becomes nearly invariant to item length across all three datasets.
Second, LBR effectively mitigates length bias during decoding, preventing the model from systematically favoring either longer or shorter items. As shown in Fig.~\ref{fig:case_study_decoding}, LBR substantially reduces the correlation between item length and the recommendation proportion. Besides, LBR yields the smallest variance in recommendation proportions across length groups, indicating that the generated recommendation distribution more closely matches that of test data. In contrast, adjusting length normalization does not effectively alleviate length bias because it does not account for differences in information content across tokens.

\subsection{Efficiency Comparison (RQ4)}
\label{sec:exp_efficiency}
In this section, we conduct efficiency comparison experiments. As shown in Fig.~\ref{fig:efficiency}, LBR maintains high efficiency while achieving the best overall performance. LBR only introduces two additional learnable parameters, \(a\) and \(b\); therefore, the extra overhead incurred in training is essentially negligible. 
In contrast, S-DPO significantly reduces training efficiency because it incorporates an additional reference model and negative samples. CFT is also inefficient, as the introduction of extra counterfactual samples doubles the training data volume. Although D3, Reweight and IGD preserve relatively good efficiency, they perform poorly in terms of accuracy.

\begin{figure}[t]
  \centering
  \includegraphics[width=\linewidth]{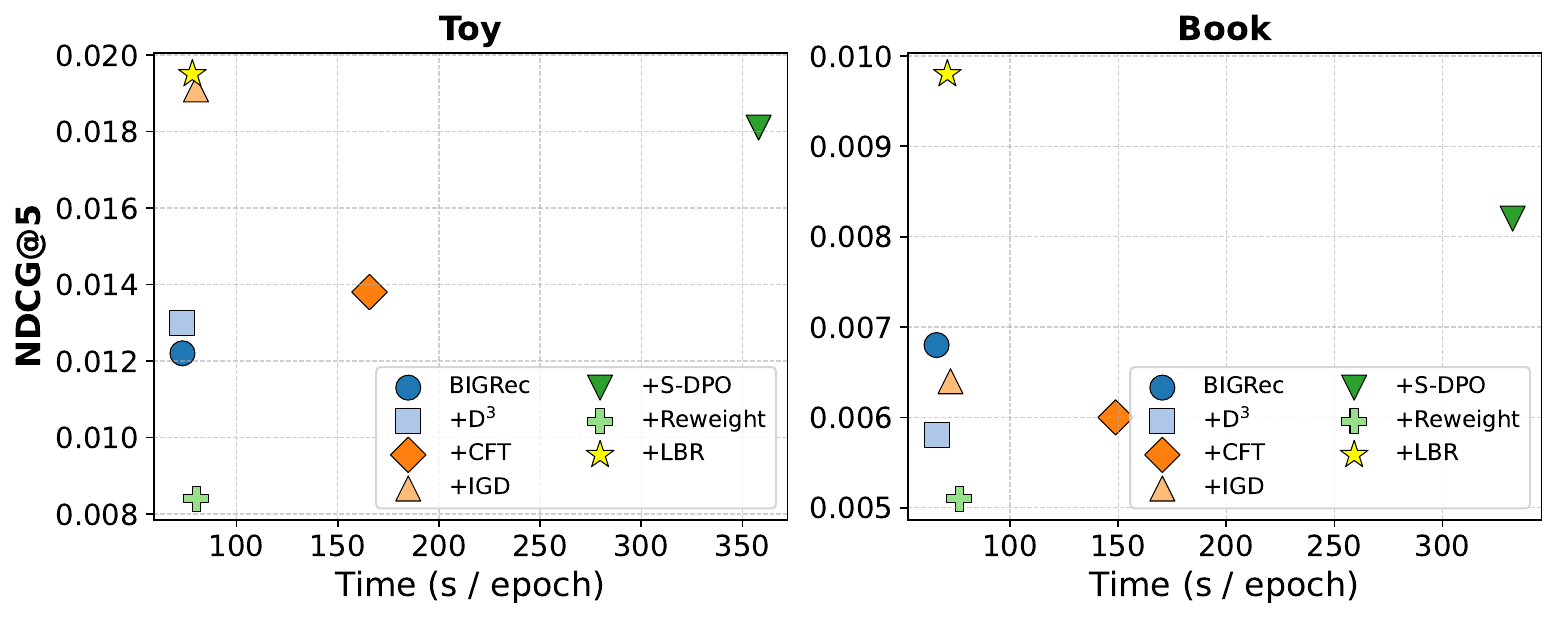}
  \caption{
    Performance comparisons in terms of both recommendation accuracy and training efficiency.
  }
  \Description{ Efficiency. }
  \label{fig:efficiency}
\end{figure}

\section{Related Work}
\subsection{Sequential Recommendation}
Sequential recommendation focuses on modeling users’ evolving preferences based on their historical interaction sequences, with the goal of predicting the next item of interest. 
With the development of deep learning, various advanced neural architectures have been integrated to model user behavior sequences, including RNN-based models such as GRU4Rec \cite{hidasi2015session} and CNN-based approaches like Caser \cite{tang2018personalized}. However, these architectures often struggle to effectively capture long-range dependencies in sequential data. More advanced models, such as SASRec \cite{kang2018self} and BERT4Rec \cite{sun2019bert4rec}, leverage self-attention mechanisms \cite{vaswani2017attention} to better model dependencies among different items in a sequence, enabling the model to automatically focus on the most relevant historical interactions for prediction. 
In addition, recent studies explore self-supervised and contrastive learning techniques to further enhance sequence representations \cite{xie2022contrastive, chen2022intent}, particularly under data sparsity and cold-start scenarios. Moreover, as user preferences may drift over time, DROS \cite{yang2023generic} introduces distributionally robust optimization \cite{sagawa2019distributionally, wang2024distributionally} to improve model robustness under distribution shifts. Readers may refer to the survey for a comprehensive review \cite{pan2026survey}.

\subsection{LLM-based Recommendation}
Large language models (LLMs), endowed with strong reasoning capabilities and extensive world knowledge \cite{dubey2024llama, achiam2023gpt, guo2025deepseek, team2024qwen2}, have been increasingly adopted in recommender systems \cite{wu2024survey, wang2025llm4dsr}. A mainstream paradigm is LLM-based recommendation, which directly employs LLMs as the recommendation backbone \cite{li2023large}. Through carefully designed prompts, the model generates recommendations conditioned on users’ historical interactions, typically represented by textual item titles. To further improve recommendation quality, recent studies have extensively explored fine-tuning LLMs on curated recommendation datasets, yielding substantial performance gains \cite{bao2025bigrec, bao2023tallrec, zhu2024collaborative, wang2025msl, wang2024flip, LLaRA, kim2024large}.

Despite the remarkable success of LLM-based recommendation, LLMs also introduce a range of bias-related challenges. Prior work has shown that LLMs can amplify biases inherent in conventional RS, such as popularity bias \cite{gao2025process, lu2025dual, jiang2024item, gao2025sprec, liao2024rosepo, lin2025recommendation}: LLMs may exhibit an even stronger tendency to recommend popular items after fine-tuning. Moreover, LLMs can introduce new forms of bias. For example, position bias \cite{chen2023bias, ma2023large, hou2024large, bito2025evaluating, jiang2025beyond, dai2024bias, luo2025recranker, zhang2024agentcf} refers to the phenomenon whereby an LLMs are influenced by an item’s position among candidate items in the prompt. Context bias \cite{wang2026doesllmfocusright} refers to the LLM’s over-reliance on auxiliary tokens when producing recommendations, hindering its ability to effectively leverage user preference-related information. Amplification bias \cite{bao2024decoding} denotes the tendency of the model to generate items composed of higher-probability tokens.
In contrast, length bias has largely been overlooked, which stems from the fact that different items may be tokenized into sequences of varying lengths. Because existing debiasing methods typically fail to account for this factor, their effectiveness remains limited. 

\subsection{Length Bias}

Prior work on length bias has mainly addressed the issue on the decoding side, largely in the context of neural machine translation \cite{stahlberg2020neural}.
Existing solutions can be broadly grouped into two lines of research. 
The first line of work mitigates length bias by modifying the decoding score. Typical approaches include length normalization and reward-based or penalty-based adjustments during beam search \cite{wu2016google,murray2018correcting,yang2018breaking,liang2022implicit}. More recent studies further propose a causal debiasing framework that treats sequence length as a confounder \cite{xuewen2021reducing}.
Another line of work targets the data level. For instance, Wang et al. \cite{provilkov2021multi} argue that length bias can stem from an overrepresentation of short sequences in the training data and propose sampling strategies that produce a more uniform length distribution.

However, in recommendation settings, the item space generated by LLMs is constrained by a Trie tree, making different tokens can convey heterogeneous amounts of information. Consequently, existing length-bias mitigation methods—which typically assume that all tokens are equally informative when computing sequence scores—are insufficient to fully eliminate the influence of length in this setting. Moreover, prior work has not considered the impact of length bias on attention.

\section{Conclusion}
This work identifies \textbf{length bias} as a fundamental issue in LLM-based recommendation, arising from the substantial variation in item token-sequence lengths and affecting both input-side self-attention and output-side Trie-Constrained decoding.
To mitigate the bias, we propose \textbf{LBR}, a lightweight and model-agnostic framework that (i) performs Length-Aware Attention Calibration (LAAC) by injecting a length-dependent offset into attention logits to reduce the spurious correlation between attention weights and item length, and (ii) introduces Effective Information Length Normalization (EILN), which leverages Trie branching factors to quantify per-token informativeness and normalizes decoding scores by an effective information length rather than raw token counts.
Extensive experiments on three real-world Amazon datasets demonstrate that LBR consistently improves both accuracy and fairness, while effectively mitigating length bias. 

This work highlights the importance of accounting for structural properties of item text under constrained generation. Future work can explore extending this perspective to broader generation constraints and recommendation settings.

\bibliographystyle{ACM-Reference-Format}
\bibliography{LBR-references}

\appendix

\section{Hyperparameter Study on $\alpha$ in Length Normalization}
\label{appsec:result_tuning_alpha}
In this section, we present experimental results on the Toy and Office datasets using different values of $\alpha$ in the length penalty, as illustrated in Fig.~\ref{fig:tuning_alpha}. As $\alpha$ increases, the model performance exhibits a consistent downward trend. The best performance is achieved when $\alpha$ is close to 0, indicating that the conventional length penalty can have a negative impact on model performance.

\begin{figure}[h]
\centering
\includegraphics[width=\linewidth]{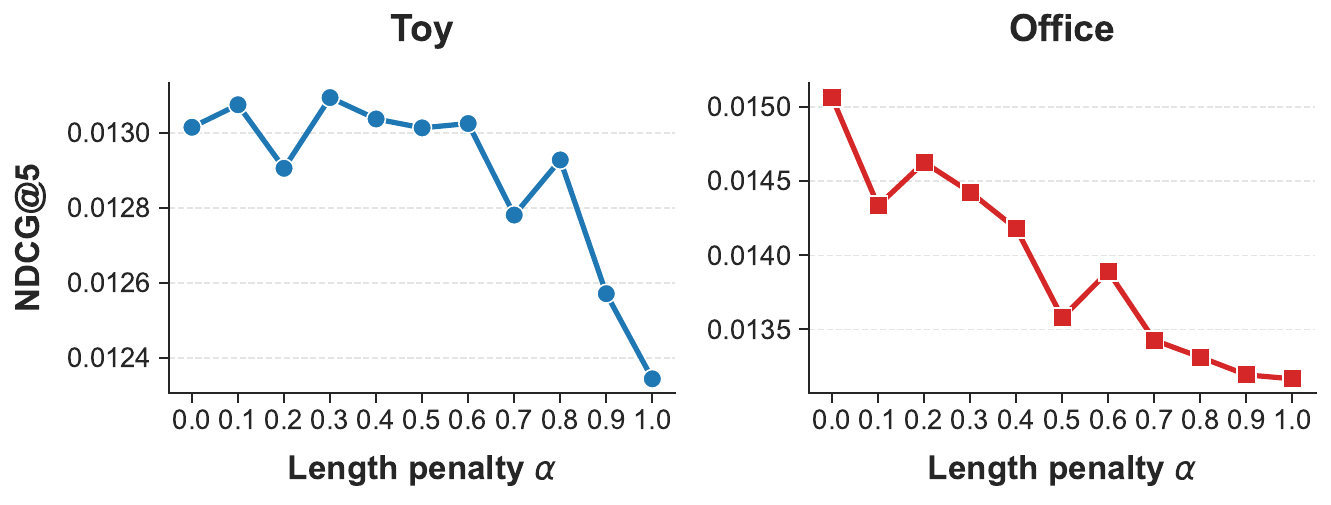}
\caption{Sensitivity analysis of the hyperparameter $\alpha$.}
\label{fig:tuning_alpha}
\end{figure}

\section{Feature Ablation Analysis of Length Bias}
\label{appsec:captum_attr_analysis}
In this section, we analyze the impact of item lengths via a feature ablation study using Captum \cite{miglani2023using} to further corroborate the length bias discussed in Section~\ref{subsec:analyze_attention}.
Specifically, we treat each item in the input prompt as an independent textual feature. By iteratively removing individual items from the prompt template, we quantify the attribution score of each historical item with respect to the generated target recommendation. As shown in Figure~\ref{fig:attribution_score}, the average attribution scores exhibit a clear monotonic increase with item length.

\begin{figure}[h]
\centering
\includegraphics[width=\linewidth]{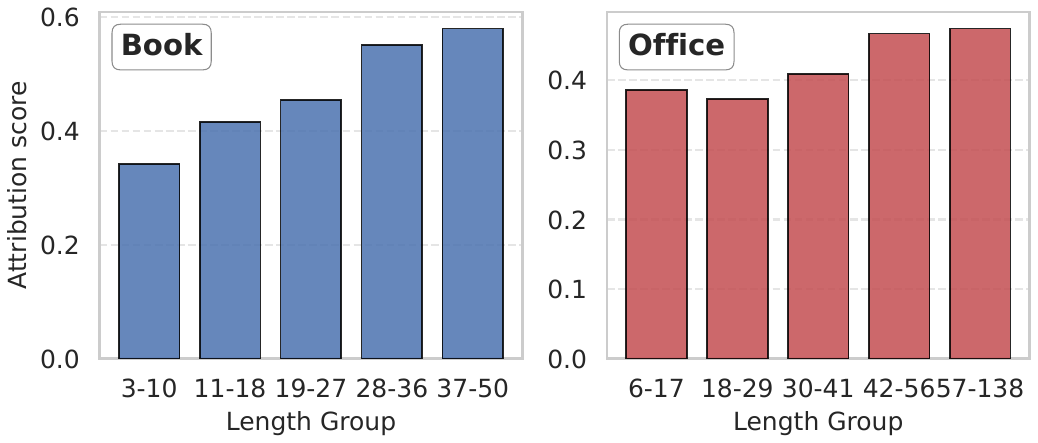}
\caption{Impact of item length on average attribution scores.}
\label{fig:attribution_score}
\end{figure}

\section{Popularity Analysis Across Item Length}
\label{appsec:popularity_analysis}

To examine whether the observed preference for shorter items is confounded by item popularity, we analyze the distribution of training interactions across item length groups. Specifically, we first sort all unique items in the item set by token length and partition them into five bins with equal numbers of items. We then assign each item appearing in the training interactions to the corresponding length bin and compute the proportion of training interactions in each bin. As shown in Figure~\ref{fig:item_popularity}, shorter items are not overrepresented in the training interactions. 

\begin{figure}[h]
\centering
\includegraphics[width=\linewidth]{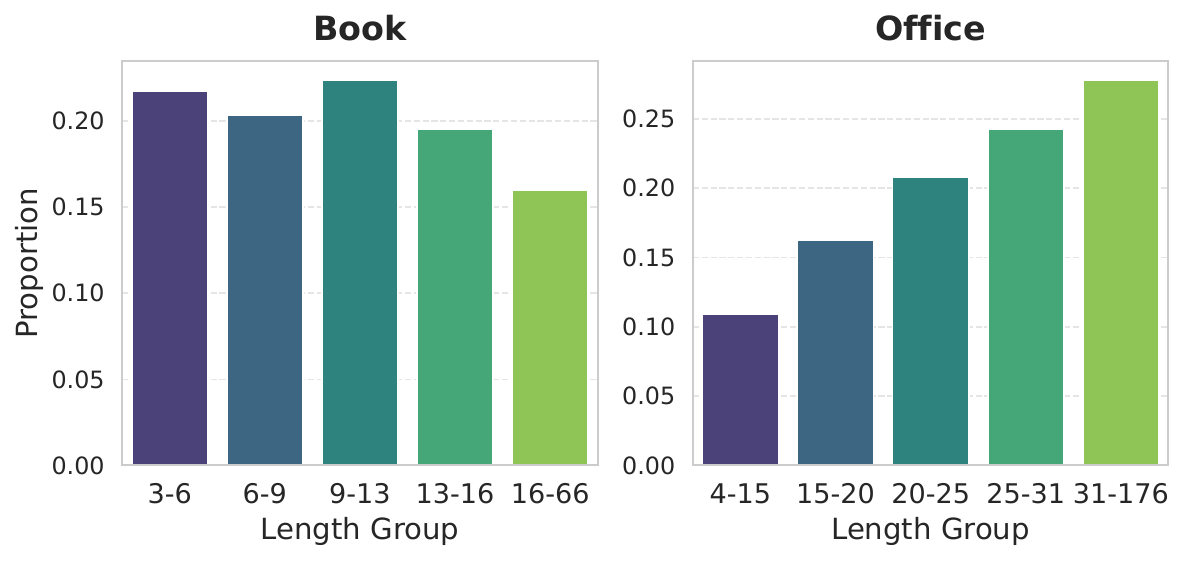}
\caption{Distribution of item length on training interactions.}
\label{fig:item_popularity}
\end{figure}

\section{Experiments with Naive Remedies for the Impact of Length Bias on Attention}
\label{appsec:result_naive}
In this section, we investigate two naive strategies to mitigate the impact of length bias on attention: (1) forcing each item's token sequence to have the same length (fixed to 4) via truncation or padding; and (2) using RQ-VAE to obtain fixed-length semantic ID for items. As shown in Table.~\ref{tab:naive_remedy}, our experiments indicate that neither preprocessing strategy yields consistent improvements; moreover, both can degrade performance in certain settings.

\begin{table}[h]
\caption{The performance results of naive methods.}
\label{tab:naive_remedy}
\begin{adjustbox}{width=\linewidth}
\begin{tabular}{@{}ccccccccc@{}}
\toprule
\multirow{2}{*}{Method}            & \multicolumn{4}{c}{Toy}                                & \multicolumn{4}{c}{Office}        \\ \cmidrule(l){2-9} 
                                   & N@5    & N@10   & H@5    & \multicolumn{1}{c|}{H@10}   & N@5    & N@10   & H@5    & H@10   \\ \midrule
\multicolumn{1}{l|}{BIGRec}        & 0.0123 & 0.0155 & 0.0209 & \multicolumn{1}{c|}{0.0305} & 0.0105 & 0.0226 & 0.0213 & 0.0587 \\
\multicolumn{1}{c|}{+Equal Length} & 0.0123 & 0.0158 & 0.0198 & \multicolumn{1}{c|}{0.0307} & 0.0060 & 0.0211 & 0.0119 & 0.0597 \\
\multicolumn{1}{c|}{+Semantic IDs} & 0.0038 & 0.0048 & 0.0056 & \multicolumn{1}{c|}{0.0088} & 0.0031 & 0.0063 & 0.0057 & 0.0156 \\ \bottomrule
\end{tabular}
\end{adjustbox}
\end{table}

\section{Binary-decision View Diagram}
\label{appsec:bin-tree}
A Trie node with branching factor $H$ can be viewed as expanding into an equivalent subtree of depth $\log_2 H$ in the binary-decision tree. In this section, we demonstrated a tree with four nodes can be represented as a two-layer binary tree (Figure~\ref{fig:bin-tree}).
\begin{figure}[h]
\centering
\includegraphics[width=0.8\linewidth]{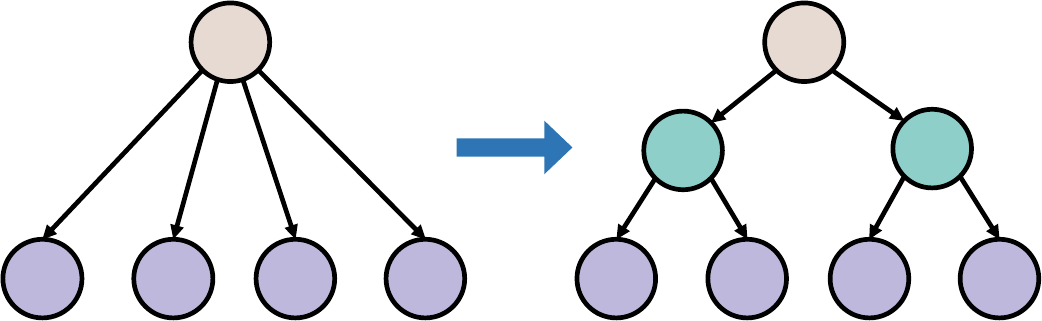}
\caption{ Binary-decision View Diagram }
\label{fig:bin-tree}
\end{figure}

\section{Theoretical proof of Length-Aware Attention Calibration}
\label{appsec:laac_invariance}
\begin{lemma}
\label{lem:laac_invariance}
Let $\mathrm{Attr}'_u(i)$ denote the calibrated item-level attention attribution (i.e., the cumulative attention mass assigned to item $i$ when attending from a query token $u$). Under mild assumptions, applying LAAC makes the expected attribution approximately invariant to item length.
\end{lemma}

\begin{proof}
Fix a query $u$. Consider $m$ historical items in the context. Item $i$ corresponds to a set of tokens $\mathcal{T}_i$ with length $l_i \triangleq |\mathcal{T}_i|$. Let $s_{uv}$ denote the pre-calibration attention logit from $u$ to a context token $v$. 
Define 
\begin{equation}Z_u(i)\triangleq \sum_{v\in\mathcal{T}_i}\exp(s_{uv}),\quad X_u(i)\triangleq \frac{Z_u(i)}{g(l_i)},\quad S_u\triangleq \sum_{j=1}^m X_u(j).\label{eq:zxs_def}
\end{equation}

We state two mild assumptions commonly used in analyses.
\paragraph{(A1)} The softmax normalizer $S_u=\sum_{j=1}^m X_u(j)$ concentrates around its mean such that, for an item $i$ of length $l$,
\begin{equation}
\label{eq:assump_a2}
\mathbb{E}\!\left[\frac{X_u(i)}{S_u}\,\bigg|\, l_i=l\right]\approx\frac{\mathbb{E}\!\left[X_u(i)\mid l_i=l\right]}{\mathbb{E}[S_u]}.
\end{equation}
In particular, this approximation is reasonable when $m$ is large and $\{X_u(j)\}_{j=1}^m$ are weakly dependent with finite variance, under which $S_u$ concentrates by the law of large numbers.
\paragraph{(A2)}
There exists a constant $\kappa_u>0$ such that for any length $l$,
\begin{equation}
\label{eq:assump_a1}
\mathbb{E}\!\left[ Z_u(i)\mid l_i=l \right] = \kappa_u\, g(l).
\end{equation}
This assumption formalizes that, before calibration, the expected attention exhibits a multiplicative length trend captured by $g(l)$. As illustrated in Fig.~\ref{fig:attention_length_bias_analysis}, there exists an approximately linear relationship between the sequence length and $Z_u(i)$, which supports the validity of this assumption.

We first express $\mathrm{Attr}'_u(i)$ in terms of the item-level masses.
\begin{align}
\label{eq:attr_ratio}
\mathrm{Attr}'_u(i)&= \frac{\sum_{v\in\mathcal{T}_i}\exp(s_{uv})/g(l_i)}{\sum_{j=1}^m\sum_{v\in\mathcal{T}_j}\exp(s_{uv})/g(l_j)} \notag\\&= \frac{Z_u(i)/g(l_i)}{\sum_{j=1}^m Z_u(j)/g(l_j)}= \frac{X_u(i)}{S_u}.
\end{align}
Conditioning on $l_i=l$ and applying Assumption (A1) in Eq.~\eqref{eq:assump_a2},
\begin{equation}
\mathbb{E}\!\left[\mathrm{Attr}'_u(i)\mid l_i=l\right]=\mathbb{E}\!\left[\frac{X_u(i)}{S_u}\,\bigg|\, l_i=l\right]\approx\frac{\mathbb{E}\!\left[X_u(i)\mid l_i=l\right]}{\mathbb{E}[S_u]}.\label{eq:mf_step}
\end{equation}
Moreover, by the definition $X_u(i)=Z_u(i)/g(l_i)$ and Assumption (A2) in Eq.~\eqref{eq:assump_a1},
\begin{equation}
\mathbb{E}\!\left[X_u(i)\mid l_i=l\right]=\frac{\mathbb{E}\!\left[Z_u(i)\mid l_i=l\right]}{g(l)}=\frac{\kappa_u g(l)}{g(l)}=\kappa_u,\label{eq:cancel}
\end{equation}
which is independent of $l$. Substituting Eq.~\eqref{eq:cancel} into Eq.~\eqref{eq:mf_step} yields
\begin{equation}\mathbb{E}\!\left[\mathrm{Attr}'_u(i)\mid l_i=l\right]\approx\frac{\kappa_u}{\mathbb{E}[S_u]}\triangleq c_u,
\end{equation}
where $c_u$ does not depend on $l$. This proves that the expected calibrated attribution is approximately length-invariant.
\end{proof}

\end{document}